\documentclass[12pt]{amsart}
\usepackage{amssymb}
\usepackage{color}
\usepackage{amsmath,epic,curves,amscd}
\usepackage[english]{babel}
\usepackage{graphicx}
\usepackage{comment}
\usepackage{appendix}
\usepackage{mathdots}
\usepackage[all]{xy}
\usepackage{mathtools}
\usepackage{lscape}
\usepackage{stmaryrd}
\usepackage{mathabx}
\newtheorem{claim}{}[section]
\newtheorem{theorem}[claim]{Theorem}

\newtheorem{lemma}[claim]{Lemma}
\newtheorem{proposition}[claim]{Proposition}
\newtheorem{corollary}[claim]{Corollary}

\theoremstyle{remark}

\renewenvironment{proof}{\noindent{\it Proof. \hskip0pt}}
                      {$\square$\par\medskip}
\allowdisplaybreaks
\usepackage{color}

\begin{document}
\baselineskip 5.9 truemm
\parindent 1.5 true pc

\newcommand\lan{\langle}
\newcommand\ran{\rangle}
\newcommand\tr{\operatorname{Tr}}
\newcommand\ot{\otimes}
\newcommand\ttt{{\text{\bf t}}}
\newcommand\rank{\ {\text{\rm rank of}}\ }
\newcommand\choi{{\rm C}}
\newcommand\dual{\star}
\newcommand\flip{\star}
\newcommand\cp{{{\mathbb C}{\mathbb P}}}
\newcommand\ccp{{{\mathbb C}{\mathbb C}{\mathbb P}}}
\newcommand\pos{{\mathcal P}}
\newcommand\tcone{T}
\newcommand\mcone{K}
\newcommand\superpos{{{\mathbb S\mathbb P}}}
\newcommand\blockpos{{{\mathcal B\mathcal P}}}
\newcommand\jc{{\text{\rm JC}}}
\newcommand\dec{{\mathbb D}{\mathbb E}{\mathbb C}}
\newcommand\decmat{{\mathcal D}{\mathcal E}{\mathcal C}}
\newcommand\ppt{{\mathcal P}{\mathcal P}{\mathcal T}}
\newcommand\pptmap{{\mathbb P}{\mathbb P}{\mathbb T}}
\newcommand\xxxx{\bigskip\par ================================}
\newcommand\join{\vee}
\newcommand\meet{\wedge}
\newcommand\ad{\operatorname{Ad}}
\newcommand\ldual{\varolessthan}
\newcommand\rdual{\varogreaterthan}
\newcommand{\slmp}{{\mathcal M}^{\text{\rm L}}}
\newcommand{\srmp}{{\mathcal M}^{\text{\rm R}}}
\newcommand{\smp}{{\mathcal M}}
\newcommand{\id}{{\text{\rm id}}}
\newcommand\tsum{\textstyle\sum}
\newcommand\hada{\Theta}
\newcommand\ampl{\mathbb A^{\text{\rm L}}}
\newcommand\ampr{\mathbb A^{\text{\rm R}}}
\newcommand\amp{\mathbb A}
\newcommand\rk{{\text{\rm rank}}\,}
\newcommand\calI{{\mathcal I}}
\newcommand\bfi{{\bf i}}
\newcommand\bfj{{\bf j}}
\newcommand\bfk{{\bf k}}
\newcommand\bfl{{\bf l}}
\newcommand\bfzero{{\bf 0}}
\newcommand\bfone{{\bf 1}}
\newcommand\pr{\prime}
\newcommand\re{{\text{\rm Re}}\,}
\newcommand\sr{{\text{\rm SR}}\,}
\newcommand\spa{{\text{\rm span}}\,}

\title{Exposedness of elementary positive maps between matrix algebras}

\author{Seung-Hyeok Kye}
\address{Department of Mathematics and Institute of Mathematics, Seoul National University, Seoul 151-742, Korea}
\email{kye at snu.ac.kr}

\thanks{partially supported by NRF-2020R1A2C1A01004587, Korea}

\begin{abstract}
The positive linear maps $\ad_s$ which send matrices $x$ to $s^*xs$ play important roles
in quantum information theory as well as matrix theory.
It was proved by Marciniak [Linear Multilinear Alg. \bf 61 \rm (2013), 970--975] that
the map $\ad_s$ generates an exposed ray of the convex cone of all positive linear maps.
In this note, we provide two alternative proofs, using Choi matrices and Woronowicz's method,
respectively.
\end{abstract}
\maketitle


\section{Introduction}

For a given $m\times n$ matrix $s$ with complex entries, the linear map $\ad_s:M_m\to M_n$ from the $m\times m$ matrix algebra $M_m$
into $M_n$ is defined by
$$
\ad_s(x)=s^*xs,\qquad x\in M_m.
$$
These maps appear typically in order related preserver problems \cite{{marcin_ban},molnar,semrl_souror} in matrix theory,
together with the compositions $\ad_s\circ\ttt$ by the transpose map $\ttt$. They
also play central roles in current quantum information theory; the convex combinations of $\ad_s$ are
completely positive maps, which are quantum channels when they preserve traces.
Especially, convex combinations of $\ad_s$ with matrices $s$ of rank one
are entanglement breaking quantum channels \cite{hsrus}, which give rise to bi-partite separable states through the
Jamio\l kowski--Choi isomorphism \cite{{choi75-10},{jam_72}}.
Entanglement breaking maps are also called superpositive maps \cite{{ando-04},{ssz}}.

After it was known \cite{yopp} that the maps $\ad_s$ and $\ad_s\circ\ttt$
generate exposed rays of the convex cone of all positive maps
when $s$ has minimal or maximal rank,
Marciniak \cite{marcin_exp} proved that this is the case for an arbitrary matrix $s$. His proof
depends on the characterization of rank one non-increasing positive maps \cite{marcin_ban}.
The purpose of this note is to provide two
alterative proofs with Choi matrices  \cite{choi75-10} and Woronowicz's method  \cite{{woronowicz-1},{woro_letter}}, respectively.

In this paper, we say that a positive linear map $\phi:M_m\to M_n$ is {\sl exposed} when it generates
an exposed ray of the convex cone $\mathbb P_1[M_m,M_n]$ of all positive linear maps from $M_m$ into $M_n$.
In the next section, we collect some preliminaries, and show that the exposedness of the identity map
is enough to confirm that the maps $\ad_s$ and $\ad_s\circ\ttt$ are
exposed for every $s$. The exposedness of the identity map will be proved in Section 3
with the corresponding Choi matrices.
We also prove it using Woronowicz's method in Section 4.

The author is grateful to Marcin Marciniak and Dariusz Chru\'{s}ci\'{n}ski for their valuable comments on the draft.


\section{Preliminaries}

Throughout this note, we will use the bilinear pairing
\begin{equation}\label{bi-lin}
\lan a,b\ran:=\tr(ab^\ttt)
\end{equation}
for $m\times n$ matrices $a$ and $b$, where $\tr$ denotes the trace.
For a linear map $\phi:M_m\to M_n$, the Choi matrix $\choi_\phi$
is defined \cite{choi75-10} by
$$
\choi_\phi=\sum_{i,j=1}^m |i\ran\lan j|\ot\phi(|i\ran\lan j|)\in
M_m\ot M_n.
$$
Then we have the identity
\begin{equation}\label{choi-rel}
\lan a\ot b,\choi_\phi\ran=\lan b,\phi(a)\ran, \qquad a\in M_m,\
b\in M_n,\ \phi\in L(M_m,M_n),
\end{equation}
where the bilinear pairings on the left-hand and right-hand sides are defined by
(\ref{bi-lin}) on matrix algebras $M_m\ot M_n$ and $M_n$, respectively.
The identity (\ref{choi-rel}) tells us that
the $(k,\ell)$ entry of the $(i,j)$ block of $\choi_\phi\in M_m(M_n)$ is given by
$\lan |k\ran\lan \ell|,\phi(|i\ran\lan j|)\ran$, when we identity $M_m\ot M_n$ with $M_m(M_n)$.

For a linear map $\phi:M_m\to M_n$, the adjoint map $\phi^*:M_n\to M_m$ is defined by
$$
\lan a,\phi^*(b)\ran=\lan\phi(a),b\ran,\qquad a\in M_m, b\in M_n.
$$
Then the Choi matrix $\choi_{\phi^*}\in M_n\ot M_m$
of the adjoint map $\phi^*$ is the flip of the Choi matrix $\choi_\phi\in M_m\ot M_n$.

A matrix $\varrho\in M_m\ot M_n$ is called separable if $\varrho$ is the nonnegative sum of rank one projections
onto product vectors in $\mathbb C^m\ot\mathbb C^n$. It turns out that $\varrho$ is separable if and only if
it is the Choi matrix of a map in the convex cone
$$
\superpos_1={\text{\rm conv}}\, \{\ad_s: \rank s=1\},
$$
where ${\text{\rm conv}}\, S$ denotes the convex hull of $S$. Maps in $\superpos_1$ are called superpositive, or
entanglement breaking. We also note that a matrix $\varrho$ in $M_m\ot M_n$ is the Choi matrix of a positive map
if an only if $\lan \zeta|\varrho|\zeta\ran\ge 0$
for every vector product $|\zeta\ran\in\mathbb C^m\ot\mathbb C^n$ \cite{jam_72}.
Such a matrix $\varrho\in M_m\ot M_n$ is called block-positive.

Considering the correspondence between
linear functionals on the tensor product of spaces and linear maps between spaces, it is natural to define \cite{eom-kye}
the bilinear pairing between tensor product $M_m\ot M_n$ and the mapping space $L(M_m,M_n)$ by
\begin{equation}\label{11}
\lan a\ot b,\phi\ran_1=\lan b,\phi(a)\ran,\qquad a\ot b\in M_m\ot M_n,\ \phi\in L(M_m,M_n).
\end{equation}
By the identity (\ref{choi-rel}), we have
$\lan \varrho,\phi\ran_1=\lan\varrho,\choi_\phi\ran$. Because every $\varrho\in M_m\ot M_n$
can be expressed by $\varrho=\choi_\psi$ for a unique $\psi\in L(M_m,M_n)$,
it is also natural to define \cite{{sko-laa}, {ssz},{gks}} the bilinear pairing between mapping spaces by
\begin{equation}\label{22}
\lan \psi,\phi\ran_2=\lan\choi_\psi,\choi_\phi\ran,\qquad \phi,\psi\in L(M_m,M_n).
\end{equation}
Since $\varrho=\choi_\psi$ implies
$\lan\varrho,\phi\ran_1=\lan\varrho,\choi_\phi\ran=\lan\choi_\psi,\choi_\phi\ran=\lan\psi,\phi\ran_2$,
we do not distinguish the bilinear pairings $\lan\,\cdot\,
,\cdot\,\ran$, $\lan\,\cdot\, ,\cdot\,\ran_1$ and $\lan\,\cdot\,
,\cdot\,\ran_2$, and use the notation $\lan\,\cdot\, ,\cdot\,\ran$
for them.

With this bilinear pairing, the convex cone $\mathbb P_1$
is dual to the convex cone $\superpos_1$. That is,
we have $\phi\in\mathbb P_1$ if and only if $\lan\phi,\psi\ran\ge 0$ for every $\psi\in\superpos_1$,
and $\psi\in\superpos_1$ if and only if $\lan\phi,\psi\ran\ge 0$ for every $\phi\in\mathbb P_1$.
Equivalently, $\phi\in\mathbb P_1$ if and only if $\lan\varrho,\phi\ran\ge 0$ for every separable $\varrho\in M_m\ot M_n$,
and $\varrho$ is separable if and only if $\lan\varrho,\phi\ran\ge 0$ for every $\phi\in\mathbb P_1$.
See the semi-expository paper \cite{kye-exposit} for more details.

For any $\phi\in\mathbb P_1$, the set $\phi^\pr=\{\psi\in\superpos_1: \lan \phi,\psi\ran=0\}$ is an exposed face of
the convex cone $\superpos_1$. The bi-dual face
$$
\phi^{\pr\pr}=\{\psi\in\mathbb P_1: \lan\varrho,\psi\ran=0\ {\text{\rm for every}}\ \varrho\in\phi^\pr\}
$$
turns out to be the smallest exposed face containing $\phi$, and so we see
\cite{eom-kye} that $\phi$ generates an exposed ray of $\mathbb P_1$
if and only if $\phi^{\pr\pr}$ consists of nonnegative scalar multiples of $\phi$.
We note that the dual object $\phi^\pr$ may be taken in $M_m\ot M_n$ or $L(M_m,M_n)$, according to the
bilinear pairing used. If we use the bilinear pairing (\ref{22}) then $\psi\in\phi^{\pr\pr}$ if and only if the following
\begin{equation}\label{bidual-2}
s\in M_{m\times n},\ \rk s=1,\ \lan \ad_s,\phi\ran=0\
\Longrightarrow\
\lan\ad_s,\psi\ran=0
\end{equation}
holds. On the other hand, $\psi\in\phi^{\pr\pr}$ is equivalent to
\begin{equation}\label{bidual-1}
a\in M_m^+,\ |\eta\ran\in\mathbb C^n,\ \lan a\ot|\eta\ran\lan\eta|,\phi\ran=0\
\Longrightarrow\
\lan a\ot|\eta\ran\lan\eta|,\psi\ran=0,
\end{equation}
when we use the bilinear pairing (\ref{11}).

By the singular value decomposition, we know that every $m\times n$ matrix is of the
form $u\sigma v^*$, where $u\in M_m$ and $v\in M_n$ are nonsingular, and $\sigma$ is given by
$$
\sigma=\sum_{i=1}^r |i\ran\lan i|\in M_{m\times n}.
$$
Because $\ad_{u\sigma v^*}=\ad_{v^*}\circ\ad_\sigma\circ\ad_{u}$ and
and the map $\phi\mapsto \ad_{v^*}\circ\phi\circ\ad_{u}$ is an
affine isomorphism between the real vector space $H(M_m,M_n)$ of all Hermiticity preserving linear maps, we see that the map
$\ad_{u\sigma v^*}$ is exposed in $\mathbb P_1[M_m,M_n]$ if and only
if $\ad_\sigma$ is exposed in $\mathbb P_1[M_m,M_n]$.

We suppose that $r\le n$, and consider the following two maps $S:M_r\to M_n$ and $T:M_n\to M_r$ defined by
\begin{equation}\label{ST-def}
S: a\mapsto \left(\begin{matrix} a&0\\0&0\end{matrix}\right)\in M_n,\qquad
T: \left(\begin{matrix} a&b\\c&d\end{matrix}\right)\mapsto a\in M_r.
\end{equation}

\begin{proposition}\label{prop-dual-exposed}
If $\phi:M_m\to M_r$ is exposed then $S\circ\phi:M_m\to M_n$ is also exposed.
\end{proposition}

\begin{proof}
Suppose that $\psi:M_m\to M_n$ belongs to $(S\circ\phi)^{\pr\pr}$.
Then for every $a\in M_m^+$ and $k=r+1,\dots,n$, we have
$\lan a\ot |k\ran\lan k|, S\circ\phi\ran=\lan |k\ran\lan k|, S\circ\phi(a)\ran=0$,
which implies $\lan |k\ran\lan k|, \psi(a)\ran=0$. Therefore, we have $S\circ T\circ\psi(a)=\psi(a)$.
In order to show that $T\circ\psi\in\phi^{\pr\pr}$, suppose that $a\in M_m$, $|\zeta\ran\in\mathbb C^r$ and
$\lan a\ot |\zeta\ran\lan\zeta|,\phi\ran=0$. Then we have $\lan |\zeta\ran\lan\zeta|, \phi(a)\ran=0$.
Writing $|\zeta^\star\ran=|\zeta\ran\oplus 0\in\mathbb C^n$,
we have $\lan|\zeta^\star\ran\lan\zeta^\star|, S\circ\phi(a)\ran=0$. Since $\psi\in (S\circ\phi)^{\pr\pr}$, we have
$\lan|\zeta^\star\ran\lan\zeta^\star|,\psi(a)\ran=0$, which implies $\lan|\zeta\ran\lan\zeta|, T\circ\phi(a)\ran=0$.
Therefore, we have $T\circ\psi\in\phi^{\pr\pr}$ and $T\circ\psi=\lambda\phi$, because $\phi$ is exposed.
This implies $\psi=S\circ T\circ\psi=\lambda (S\circ\phi)$, as it was required.
\end{proof}

\begin{center}
\setlength{\unitlength}{0.3 truecm}
\begin{picture}(10,10)
\put (0,8){$M_m$}
\put (2.5,8){\vector (4,-3){8}}
\put (11,1){$M_n$}
\put (2.5,8.5){\vector(1,0){8}}
\put (11,8){$M_r$}
\put (12,7.5){\vector(0,-1){5.1}}
\put (11.5,2.4){\vector(0,1){5,1}}
\put (6.4,8.9){$\phi$}
\put (5.7,4.0){$\psi$}
\put (12.2,4.9){$S$}
\put (10.4,4.9){$T$}
\end{picture}
\end{center}

Suppose that the identity map $\id_r:M_r\to M_r$ is shown to be exposed. Then we have $S=S\circ\id_r$
is exposed by Proposition \ref{prop-dual-exposed}. Since
$\phi\mapsto\phi^*$ is an affine isomorphism between $H(M_r,M_n)$ and $H(M_n,M_r)$, we see that $T=S^*$ is also exposed.
Now, we consider the exposed map $T:M_m\to M_r$ with the same definition as in (\ref{ST-def}).
Then we may conclude that $\ad_\sigma=S\circ T$ is exposed by Proposition \ref{prop-dual-exposed} again.
In the next two sections, we prove that the identity map $\id_r$ is exposed.
We also note that $\phi\mapsto \phi\circ\ttt$ gives rise to an affine isomorphism of $H(M_m,M_n)$, and so
exposedness of $\phi:M_m\to M_m$ implies that of $\phi\circ\ttt:M_m\to M_n$.


\section{A proof with Choi matrices}

In this section, we show that the identity map $\id_r$ is exposed. To see this, we suppose that
$\psi\in \mathbb P_1[M_r,M_r]$ belongs to $\id_r^{\pr\pr}$. Then we have
$$
s\in M_{r},\ \rk s=1, \ \lan \id_r,\ad_s\ran=0\
\Longrightarrow\ \lan\psi,\ad_s\ran=0
$$
by (\ref{bidual-2}). We are going to show that $\psi$ is a nonnegative scalar
multiple of $\id_r$.

We will identify $M_r\ot M_r$ with
$M_r(M_r)=M_{r^2}$, and the entries of $\choi_\psi\in M_{r^2}$ will be
denoted by $c_{(i,k),(j,\ell)}$ with $i,j,k,\ell=1,\dots, r$, where $(i,k)$'s and $(j,\ell)$'s are endowed
with the lexicographic orders. Therefore,
$$
c_{(i,k),(j,\ell)} =\lan |ik\ran\lan j\ell|,\choi_\psi\ran =\lan
|i\ran\lan j|\ot |k\ran\lan \ell|,\choi_\psi\ran =\lan |k\ran\lan
\ell|, \psi(|i\ran\lan j|)\ran
$$
is the $(k,\ell)$ entry of the $(i,j)$ block of $\choi_\psi\in
M_r(M_r)$.
When $J$ is a subset of $\{1,2,\dots,r\}\times \{1,2,\dots,r\}$, we
denote by $\varrho_J$ the principal submatrix of $\varrho\in M_r(M_r)$ by
taking $(i,k)$ rows and columns for $(i,k)\in J$.

We begin with the identity map $\id_2$ between $2\times 2$ matrices.
Note that the Choi matrix of $\id_2$ is given by
\begin{equation}\label{choi_id}
\choi_{\id_2}=
\left(\begin{matrix}
1 &0 &0 &1\\
0 &0 &0 &0\\
0 &0 &0 &0\\
1 &0 &0 &1
\end{matrix}\right).
\end{equation}
We first consider block-positive matrices in $M_2\ot M_2$ whose diagonals contain zero entries.
Suppose that $(1,2)$-th and $(2,1)$-th diagonals of $\varrho\in M_2\ot M_2$ are zero.
Then by the definition of
block-positivity, the $2\times 2$ principal submatrices
$\varrho_{\{(1,1),(1,2)\}}$, $\varrho_{\{(2,1),(2,2)\}}$,
$\varrho_{\{(1,1),(2,1)\}}$ and $\varrho_{\{(1,2),(2,2)\}}$ are
positive, and so $\varrho$ must be of the form
\begin{equation}\label{choi-phi-ff}
\varrho_{a,b,\alpha,\beta}=
\left(\begin{matrix}
a &0 &0 &\alpha\\
0 &0 &\beta &0\\
0 &\bar\beta &0 &0\\
\bar\alpha &0 &0 &b
\end{matrix}\right).
\end{equation}
The following is a special case of multi-qubit cases \cite[Theorem 5.5]{han_kye_multi},
but we include a simple proof for the convenience of readers.

\begin{lemma}\label{2x2_block_pos}
Suppose that $a,b\ge 0$ and $\alpha,\beta\in\mathbb C$. Then
$\varrho_{a,b,\alpha,\beta}\in M_2(M_2)$ is block-positive if and
only if $|\alpha|+|\beta|\le \sqrt{ab}$.
\end{lemma}

\begin{proof}
Suppose that $\varrho_{a,b,\alpha,\beta}$ is block-positive, and
take a product vector
$$
|\xi\ran=(p, pe^{{\rm i}\tau}, qe^{{\rm i}(\theta-\tau)}, qe^{{\rm i}\theta})^\ttt\in \mathbb C^2\ot\mathbb C^2
$$
with real numbers $p,q$. Then we have
$$
0\le \lan\xi|\varrho_{a,b,\alpha,\beta}|\xi\ran
=p^2a+q^2b +2pq \re\left[ e^{{\rm i}\theta}(\alpha+\beta e^{-2{\rm i}\tau})\right]
$$
for every $p,q\in\mathbb R$, and it follows that
$\left|\re\left[e^{{\rm i}\theta}(\alpha+\beta e^{-2{\rm i}\tau})\right]\right|^2\le ab$
for every $\theta$ and $\tau$. Therefore, we have the required condition.
Under the condition, it is easy to see that $\varrho_{a,b,\alpha,\beta}$
is the Choi matrix of a positive decomposable map.
\end{proof}

\begin{corollary}\label{lem-exp_1}
Suppose that a block-positive $\varrho\in M_2\ot M_2$ has at most one
nonzero diagonal entry. Then all the other entries of $\varrho$ are
zero.
\end{corollary}


\begin{proposition}\label{lem-exp_2}
The identity map $\id_2$ on $M_2$ generates an exposed ray in $\mathbb
P_1[M_2,M_2]$.
\end{proposition}

\begin{proof}
Suppose that $\psi\in\id_2^{\pr\pr}$. Taking $s=|i\ran\lan j|\in
M_{2\times 2}$ with $i\neq j$ in (\ref{bidual-2}), we see that
$\choi_\psi$ must be of the form $\varrho_{a,b,\alpha,\beta}$. We
also take $s=\left(\begin{matrix}1&e^{{\rm i}\theta}\\-e^{-{\rm
i}\theta}&-1\end{matrix}\right)$ in (\ref{bidual-2}), to see
that
$$
0=\lan \choi_{\ad_s}, \varrho_{a,b,\alpha,\beta}\ran
=a+b-2\re(\alpha+\beta e^{-2{\rm i}\theta}),
$$
for every $\theta$. Therefore, we have $\beta=0$, and we also have
$$
0= a+b-2\re\alpha \ge a+b -2|\alpha|\ge 2\sqrt {ab} -2|\alpha|\ge 0,
$$
by Lemma \ref{2x2_block_pos}. This happens only when $a=b=\alpha$.
\end{proof}

In order to prove that $\id_r$ is exposed in $\mathbb
P_1[M_r,M_r]$, we suppose that $\psi\in\id_r^{\pr\pr}$.
Considering the $r\times r$ matrix $s=|i\ran\lan k|$ in
(\ref{bidual-2}), we see that the diagonal entries of
$\choi_\psi$ are given by
\begin{equation}\label{exp_diago-ent}
c_{(i,k),(i,k)}=0,\qquad i\neq k,\ i,k=1,2,\dots,r.
\end{equation}
In order to determine off-diagonal entries $c_{(i,k),(j,\ell)}$ of $\choi_\psi$
with $(i,k)\neq (j,\ell)$, we consider the principal submatrix of $[\choi_\psi]_J$ of $\choi_\psi$, with
\begin{equation}\label{4x4-prib-sub}
J=\{(i,k), (i,\ell), (j,k), (j,\ell)\}\subset \{1,\dots,r\}\times\{1,\dots,r\}.
\end{equation}
If $i=j$ or $k=\ell$ then $[\choi_\psi]_J$ is a $2\times 2$ positive matrix and so $c_{(i,k),(j,\ell)}=0$
by the diagonal entries in (\ref{exp_diago-ent}). Otherwise, $[\choi_\psi]_J$ is a
block matrix in $M_2\ot M_2=M_2(M_2)$. We will show that this
is block-positive.

For this purpose, we consider the linear map $\lambda_{i,j}:M_2\to M_r$ defined by
$$
\lambda_{i,j} :\left(\begin{matrix}
a_{11}&a_{12}\\a_{21}&a_{22}\end{matrix}\right) \mapsto
a_{11}|i\ran\lan i| +a_{12}|i\ran\lan j|+a_{21}|j\ran\lan
i|+a_{22}|j\ran\lan j|\in M_m,
$$
for $i,j=1,\dots,r$ with $i\neq j$.
Then the adjoint map $(\lambda_{k,\ell})^*:M_m\to M_2$ is given by
$$
(\lambda_{k,\ell})^* :\sum_{i,j=1}^r a_{ij}|i\ran\lan j|\mapsto
\left(\begin{matrix} a_{kk}&a_{k\ell}\\ a_{\ell
k}&a_{\ell\ell}\end{matrix}\right)\in M_2.
$$
The composition $(\lambda_{k,\ell})^*\circ\psi\circ\lambda_{i,j}$ is a positive map from
$M_2$ into $M_2$, whose Choi matrix is given by
$[\choi_\psi]_J$ with $J$ in (\ref{4x4-prib-sub}). Therefore,  $[\choi_\psi]_J$ is
block-positive.

We see by Corollary \ref{lem-exp_1} that $c_{(i,k),(j,\ell)}\neq 0$ only when $J$ in (\ref{4x4-prib-sub})
contains at least two nonzero diagonal entries of $[\choi_\psi]_J$ .
By diagonal entries given in (\ref{exp_diago-ent}), this happens
only for $c_{(i,i),(j,j)}$ and $c_{(i,j),(j,i)}$ with
$i,j=1,2,\dots,r$ and $i\neq j$. In this case, $J$ is given by
\begin{equation}\label{4x4_ij}
J=\{(i,i),(i,j),(j,i),(j,j)\},
\end{equation}
and the corresponding principal submatrix of $\choi_{\id_r}$
is given by $\choi_{\id_2}$ in (\ref{choi_id}) which is the
Choi matrices of the map $(\lambda_{i,j})^*\circ\id_r\circ\lambda_{i,j}$. On the other hand,
$[\choi_\psi]_J$ is of the form $\varrho_{a,b,\alpha,\beta}$ in (\ref{choi-phi-ff}),
which is the Choi matrix of $(\lambda_{i,j})^*\circ\psi\circ\lambda_{i,j}$.

Now, we show that
$(\lambda_{k,\ell})^*\circ\psi\circ\lambda_{i,j}$ belongs to
$((\lambda_{k,\ell})^*\circ\id_r\circ\lambda_{i,j})^{\pr\pr}$
in $\mathbb P_1[M_2,M_2]$. To see this, take a matrix $s\in
M_{2\times 2}$ of rank one satisfying the relation $\lan
(\lambda_{k,\ell})^*\circ\id_r\circ\lambda_{i,j},
\ad_s\ran=0$. Then we have  $\lan \id_r,
\lambda_{k,\ell}\circ\ad_s\circ(\lambda_{i,j})^*\ran=0$. Furthermore, we see that
$\lambda_{k,\ell}\circ\ad_s\circ(\lambda_{i,j})^*$ belongs to
$\superpos_1$, because
$\lambda_{k,\ell}\circ\ad_s\circ(\lambda_{i,j})^*=\ad_{\hat s}$,
and
$$
\hat s=s_{11}|i\ran\lan k|+ s_{12}|i\ran\lan \ell| +s_{21}|j\ran\lan
k|+ s_{22}|j\ran\lan \ell|\in M_{m\times n}
$$
is of rank one. From $0= \lan \lambda_{k,\ell}\circ\ad_s\circ(\lambda_{i,j})^*,
\psi\ran =\lan \ad_s, (\lambda_{k,\ell})^*\circ\psi\circ
\lambda_{i,j}\ran$, we see that $(\lambda_{k,\ell})^*\circ\psi\circ \lambda_{i,j}$
belongs to $((\lambda_{k,\ell})^*\circ\id_r\circ
\lambda_{i,j})^{\pr\pr}$ in $\mathbb P_1[M_2,M_2]$.

Proposition \ref{lem-exp_2} tells us that the Choi matrix $[\choi_\psi]_J$
of $(\lambda_{i,j})^*\circ\psi\circ \lambda_{i,j}$ with $J$ in (\ref{4x4_ij})
must be a nonnegative scalar multiple of $\choi_{\id_2}$. Because each diagonal entry of $\choi_\psi$
appears in $[\choi_\psi]_J$ and $[\choi_\psi]_K$ for different $J$ and $K$,
we see that the whole matrix $\choi_\psi$ is a nonnegative scalar multiple
of $\choi_{\id_r}$. This completes the proof that the identity map $\id_r$ is exposed. Therefore, we have the following:

\begin{theorem}{\rm (Marciniak \cite{marcin_exp})}\label{ad-exposed}
For every $s\in M_{m\times n}$, the maps $\ad_s:M_m\to M_n$ and
$\ad_s\circ\ttt:M_m\to M_n$ generate exposed rays of the convex cone $\mathbb P_1[M_m,M_n]$.
\end{theorem}


\section{A proof with Woronowicz's method}

Suppose that $\phi:M_m\to M_n$ is a positive linear map.
The main idea of Woronowicz \cite{woro_letter} to show that $\phi$ is exposed is
to define the linear map $\hat\phi:M_m\ot\mathbb C^n\to\mathbb C^n$ by
$$
\hat\phi(a\ot |\eta\ran)=\phi(a)|\eta\ran,\qquad a\in M_m,\ \eta\in\mathbb
C^n,
$$
and the subspace $N_\phi$ in $M_m\ot\mathbb C^n$ by
$$
N_\phi=\spa \{a\ot |\eta\ran\in M_m^+\ot\mathbb C^n:
\phi(a)|\eta\ran=0\}.
$$
See also \cite{woronowicz-1}. In general, we have $N_\phi\subset \ker\hat\phi$.
By the identity
$$
\lan \eta|\phi(a)|\eta\ran=\lan | \bar\eta\ran\lan\bar\eta|,
\phi(a)\ran =\lan a\ot |\bar\eta\ran\lan\bar\eta|, \phi\ran,
$$
we see that
$$
a\ot |\bar\eta\ran\lan\bar\eta| \in \phi^\prime \ \Longleftrightarrow\
a\ot |\eta\ran\in N_\phi
$$
holds for $a\in M_m^+$ and $|\eta\ran\in\mathbb C^n$.
Therefore, we see that $\psi\in\phi^{\pr\pr}$ if and only if (\ref{bidual-1}) holds if and only if
$N_\phi\subset N_\psi$. If $\phi$ satisfies the condition
\begin{equation}\label{woro-con}
\ker\hat\phi=N_\phi,
\end{equation}
then we have $\ker\hat\phi=N_\phi\subset N_\psi\subset \ker\hat\psi$. Therefore, we conclude that there exists a linear map
$X:\mathbb C^n\to \mathbb C^n$ such that $\hat\psi=X\circ\hat\phi$, or equivalently, there exists $X\in M_n$ such that
$$
\psi(a)|\eta\ran=X\phi(a)|\eta\ran,\qquad a\in M_n,\ |\eta\ran\in\mathbb C^m.
$$
We recall that a positive map $\phi$ is called irreducible if the range of $\phi$ has the trivial commutant.
Suppose that $\phi$ is unital and irreducible. Then for every $a\in M_m$, we have
$X\phi(a)=\psi(a)=\psi(a)^*=\phi(a)X^*$. Taking $a=I$, we have $X=X^*$, and $X$ is in the commutant of the range.
Therefore, we have the following:

\begin{theorem}{\rm (Woronowicz \cite{woro_letter})}\label{woro-th}
Suppose that a unital irreducible positive linear map $\phi: M_m\to M_n$ satisfies {\rm (\ref{woro-con})}.
Then $\phi$ is exposed.
\end{theorem}

We will show that the identity map $\id_r:M_r\to M_r$ satisfies the condition (\ref{woro-con}).
To do this, we put
\begin{equation}\label{xi}
|\xi\ran=(1, \alpha_2,\dots,\alpha_i,\dots,\alpha_r)^\ttt\in\mathbb C^r.
\end{equation}
We fix $i=2,3,\dots,r$ for a moment, and we also put
\begin{equation}\label{eta}
|\eta_i\ran=(\bar\alpha_i,0,\dots,0,-1,0,\dots,0)^\ttt\in\mathbb C^r,
\end{equation}
where $-1$ is at the $i$-entry. Then we have $|\xi\ran\lan\xi|\ot|\eta_i\ran\in N_{\id_r}$. We identify
\begin{equation}\label{V_i}
|\xi\ran\lan\xi|\ot|\eta\ran \ \leftrightarrow\ |\xi\ran\ot|\bar\xi\ran\ot|\eta\ran
\end{equation}
between $M_r^+\ot\mathbb C^r$ and $\mathbb C^r\ot\mathbb C^r\ot\mathbb C^r$ to count dimensions.
For each $i=2,\dots,r$, we define the subspace $V_i\subset \mathbb C^r\ot\mathbb C^r\ot\mathbb C^r$ by
$$
V_i=\spa\{ |\xi\ran \ot |\bar\xi\ran\ot |\eta_i\ran: \alpha_2,\alpha_3,\dots,\alpha_r\in\mathbb C\}.
$$
Because monomials with variables $1$, $\alpha_i$ and $\bar\alpha_i$ are linearly independent \cite{ha+kye_woro_exam},
the dimension of $V_i$ coincides with the number of monomials appearing in the entries of $|\xi\ran\ot|\bar\xi\ran\ot |\eta_i\ran$.
The entries of $|\bar\xi\ran\ot |\eta_i\ran$ are given by the following $2r-1$ monomials
\begin{equation}\label{entries}
\begin{matrix}
\bar\alpha_i, &\bar\alpha_i\bar\alpha_2,&\dots,&\bar\alpha_i^2,&\dots,&\bar\alpha_i\bar\alpha_r,\\
-1, &-\bar\alpha_2,&\dots,&-\bar\alpha_i,&\dots,&-\bar\alpha_r.
\end{matrix}
\end{equation}
Therefore, we see that
$$
\dim V_i=r(2r-1), \qquad i=2,\dots,r.
$$
By the entries (\ref{entries}) of $|\bar\xi\ran\ot |\eta_i\ran$,
we see that the orthogonal complement $V_i^\perp$ is spanned by orthogonal vectors
\begin{equation}\label{vector}
\begin{aligned}
&|j11\ran+|jii\ran,\qquad j=1,2,\dots,r,\\
&|xyz\ran,\qquad x,y=1,2,\dots,r,\quad z=2,\dots,i-1,i+1,\dots,r,
\end{aligned}
\end{equation}
whose cardinality is $r+r^2(r-2)=r^3-r(2r-1)$.
If $r=2$ then we have $\dim N_{\id_2}=6$.

When $r\ge 3$, it is worthwhile to write down the generators of $V_i^\perp$ as follows:
$$
\begin{aligned}
&V_2^\perp: & |111\ran+|122\ran,\,|211\ran+|222\ran,\dots, &|r11\ran+|r22\ran,\ &|xy3\ran,\, &|xy4\ran,\cdots,|xyr\ran,\\
&V_3^\perp: & |111\ran+|133\ran,\,|211\ran+|233\ran,\dots, &|r11\ran+|r33\ran,\ &|xy2\ran,\, &|xy4\ran,\cdots,|xyr\ran,\\
&\cdots\\
&V_r^\perp: & |111\ran+|1rr\ran,\,|211\ran+|2rr\ran,\dots, &|r11\ran+|rrr\ran,\ &|xy2\ran,\, &|xy3\ran,\cdots,|xy(r-1)\ran,\\
\end{aligned}
$$
to see that $V_j^\perp\cap V_k^\perp=V_2^\perp\cap V_3^\perp\cap\dots\cap V_r^\perp$ is spanned by the following $r$ vectors
$$
\begin{aligned}
|\zeta_1\ran:=&|111\ran+|122\ran+|133\ran+\cdots+|1rr\ran,\\
|\zeta_2\ran:=&|211\ran+|222\ran+|233\ran+\cdots+|2rr\ran,\\
\cdots\\
|\zeta_r\ran:=&|r11\ran+|r22\ran+|r33\ran+\cdots+|rrr\ran,
\end{aligned}
$$
for $j,k=2,3,\dots,r$.
So far, we have seen that $N_{\id_r}= \spa\{V_2, V_3,\dots, V_r\}$ and
$$
N_{\id_r}^\perp= V_2^\perp\cap V_2^\perp\cap\dots\cap V_r^\perp =\spa\{|\zeta_1\ran,\dots,|\zeta_r\ran\}.
$$
Therefore, we have $\dim N_{\id_r}=r^3-r$ which coincides with the dimension of $\ker\widehat{\id_r}$.
By Theorem \ref{woro-th}, we conclude that the identity map is exposed.

It turns out that the maps $S:M_r\to M_n$ and $T:M_m\to M_r$ defined in (\ref{ST-def})
do not satisfy the conditions of Theorem \ref{woro-th} unless $m=n=r$. In the remainder of this note,
we show that Woronowicz's argument is still useful to show the exposedness of the maps $S$ and $T$,
without Proposition \ref{prop-dual-exposed}. See also \cite{chru--W}.

With a modification of the above argument, one may show that the map $S$ also satisfies the condition (\ref{woro-con}).
We modify $|\eta_i\ran$ in (\ref{eta}) and $V_i$ in (\ref{V_i}) by
$$
\begin{aligned}
|\eta_i\ran&=(\bar\alpha_i,0,\dots,0,-1,0,\dots,0,\beta_{r+1},\dots,\beta_n)\in\mathbb C^n,\\
V_i&=\spa\{|\xi\ran\ot|\bar\xi\ran\ot|\eta_i\ran: \alpha_2,\dots,\alpha_r,\beta_{r+1},\dots,\beta_n\}
    \subset\mathbb C^r\ot\mathbb C^r\ot\mathbb C^n.
\end{aligned}
$$
Then $\dim V_i=r[(2+n-r)r-1]$ and $V_i^\perp$ is spanned by the vectors (\ref{vector}) in $\mathbb C^r\ot\mathbb C^r\ot\mathbb C^n$.
We also have $\dim N_S^\perp=r$, and so $\dim N_S=r^2n-r$. This coincides with the dimension of $\ker \hat S$, and so
the map $S$ satisfies (\ref{woro-con}). Therefore, we see that
there exists $X\in M_n$ such that $\psi(a)=XS(a)$ for any $\psi\in S^{\pr\pr}$. Even though $S$ is not irreducible and $X$ itself need not
to be a scalar matrix, one see that $\psi$ is a scalar multiple of $S$ from the identity $\psi(a)=X S(a)$.

It should be noted that the map $T:M_m\to M_r$ does not satisfy the condition (\ref{woro-con}) even though $T$ is exposed.
In this case, $|\xi\ran$ in (\ref{xi}) is a vector in $\mathbb C^m$, and $V_i\subset \mathbb C^m\ot\mathbb C^m\ot\mathbb C^r$
with $\dim V_i=m(2m-1)$. The orthogonal complement $V_i^\perp$ is spanned by
$$
\begin{aligned}
&|j11\ran+|jii\ran,\qquad j=1,2,\dots,m,\\
&|xyz\ran,\qquad x,y=1,2,\dots,m,\ z=2,\dots,i-1,i+1,\dots,r
\end{aligned}
$$
whose cardinality is $m+m^2(r-2)=m^2r-m(2m-1)$. The main difference occurs in
the space $V_2^\perp\cap V_3^\perp\cap\dots\cap V_r^\perp$, which is spanned by
$$
|\zeta_k\ran:=|k11\ran+|k22\ran+|k33\ran+\cdots+|krr\ran,\qquad k=1,2,\dots m.
$$
Therefore, we have
$$
\dim N_T=m^2r-m,\qquad \dim\ker\hat T=m^2r-r.
$$
Note that $|\zeta_k\ran$ belongs to $\ker\hat T \ominus N_T$ for $k=r+1,r+2,\dots,m$.
When $\psi\in T^{\pr\pr}$, one may show that $|\zeta_k\ran\in \ker\hat\psi$ for $k=r+1,r+2,\dots,m$.
Then we have $\ker \hat T\subset\ker \hat\psi$, and there exists $X\in M_r$ such that $\psi(a)=XT(a)$ for any $\psi\in T^{\pr\pr}$.
Since $T$ is unital irreducible, we conclude that $\psi$ is a scalar multiple of the map $T$.

\end{document}